\documentclass{easychair}

\usepackage{doc}

\usepackage{graphicx}
\usepackage{amsmath}
\usepackage{amsthm}
\usepackage{booktabs}
\usepackage{algorithm}
\usepackage[noend]{algorithmic}
\usepackage{adjustbox}

\usepackage{tikz}
\usetikzlibrary{positioning, arrows, arrows.meta, trees, shapes.geometric}
\usepackage{environ}

\def\prob {$\mathsf{PROB}$}

\newtheorem{theorem}{Theorem}

\def\probroute {$\mathsf{ProbRoute}$}
\def\probinfer {$\mathsf{ProbInfer}$}
\def\probinfer {$\mathsf{ComputeProb}$}
\def\probsample {$\mathsf{ProbSample}$}
\def\problearn {$\mathsf{ProbLearn}$}

\newtheorem{definition}{Definition}
\newtheorem{proposition}{Proposition}
\newtheorem{property}{Property}
\newtheorem*{remark*}{Remark}
\newtheorem{remark}{Remark}

\newcommand{\lo}[1]{\mathsf{Lo}(#1)}
\newcommand{\hi}[1]{\mathsf{Hi}(#1)}
\newcommand{\locounter}[1]{\mathsf{{Lo}_{\#}}(#1)}
\newcommand{\hicounter}[1]{\mathsf{{Hi}_{\#}}(#1)}
\newcommand{\child}[1]{\mathsf{Child}(#1)}
\newcommand{\parent}[1]{\mathsf{Parent}(#1)}
\newcommand{\varset}[1]{\mathsf{VarSet}(#1)}
\newcommand{\var}[1]{\mathsf{Var}(#1)}
\newcommand{\subdiagram}[1]{\mathsf{Subdiagram}(#1)}
\newcommand{\simpletrip}[1]{\mathsf{SimpleTrip}(#1)}
\newcommand{\sol}[1]{\mathsf{Sol}(#1)}
\newcommand{\rep}[2]{\mathsf{Rep_{#1}}(#2)}
\newcommand{\m}[1]{\mathcal{M}(#1)}
\newcommand{\invm}[1]{\mathcal{M'}(#1)}
\newcommand{\terminal}[1]{\mathsf{Term}(#1)}

\newcommand{\hiparam}[1]{\mathsf{\theta_{Hi}}(#1)}
\newcommand{\loparam}[1]{\mathsf{\theta_{Lo}}(#1)}

\newcounter{term}

\title{Scalable Probabilistic Routes}

\author{
Suwei Yang\inst{1,2}
\and
Victor C. Liang\inst{2}
\and
Kuldeep S. Meel\inst{1}
}

\institute{
  National University of Singapore,
  Singapore
\and
  GrabTaxi Holdings,
  Singapore\\
}

\authorrunning{Yang, Liang and Meel}

\titlerunning{Scalable Probabilistic Routes}

\begin{document}

\maketitle

\begin{abstract}
  Inference and prediction of routes have become of interest over the past decade owing to a dramatic increase in package delivery and ride-sharing services. Given the underlying combinatorial structure and the incorporation of probabilities, route prediction involves techniques from both formal methods and machine learning. One promising approach for predicting routes uses decision diagrams that are augmented with probability values. 
  However, the effectiveness of this approach depends on the size of the compiled decision diagrams.
  The scalability of the approach is limited owing to its empirical runtime and space complexity.
  In this work, our contributions are two-fold: first, we introduce a relaxed encoding that uses a linear number of variables with respect to the number of vertices in a road network graph to significantly reduce the size of resultant decision diagrams. Secondly, instead of a stepwise sampling procedure, we propose a single pass sampling-based route prediction. In our evaluations arising from a real-world road network, we demonstrate that the resulting system achieves around twice the quality of suggested routes while being an order of magnitude faster compared to state-of-the-art.
\end{abstract}

\section{Introduction} \label{sec:introduction}

The past decade has witnessed an unprecedented rise of the service economy, best highlighted by the prevalence of delivery and ride-sharing services~\cite{ride-hailing-rise,food-delivery-rise}. For operational and financial efficiency, a fundamental problem for such companies is the inference and prediction of routes taken by the drivers. 
When a driver receives a job to navigate from location A to B, the ride-sharing app needs to predict the route in order to determine: (1) the trip time, which is an important consideration for the customer, (2) the fare, important consideration for both the driver and the customer, and (3) the trip experience since customers feel safe when the driver takes the route described in their app~\cite{ride-sharing-pricing,ride-sharing-eta}. 
However the reality is that drivers and customers have preferences, as such the trips taken are not always the shortest possible by distance or time~\cite{ride-sharing-route-preference}.
To this end, delivery and ride-sharing service companies have a need for techniques that can infer the distribution of routes and efficiently predict the likely route a driver takes for a given start and end location. 

Routing, a classic problem in computer science, has traditionally been approached without considering the learning of distributions~\cite{shortest-path-problem,tsp-problem}. However, Choi, Shen, and Darwiche demonstrated through a series of papers that the distribution of routes can be conceptualized as a structured probability distribution (SPD) given the underlying combinatorial structure~\cite{psdd-route,conditional-psdd-sbn,structured-baysian-networks}. Decision diagrams, which are particularly well-suited for representing SPDs, have emerged as the state-of-the-art approach for probability guided routing. The decision diagram based approach allows for learning of SPDs through the use of decision diagrams augmented with probability values, followed by a stepwise process for uncovering the route node by node.

However, scalability remains a challenge when using decision diagrams to reason about route distributions, particularly for large road networks. Existing works address this concern in various ways, such as through the use of hierarchical diagrams~\cite{psdd-route} and Structured Bayesian Networks~\cite{conditional-psdd-sbn}. Choi et al.~\cite{psdd-route} partition the structured space into smaller subspaces, with each subspace's SPD being represented by a decision diagram. Shen et al. used Structured Bayesian Networks to represent conditional dependencies between sets of random variables, with the distribution within each set of variables represented by a conditional Probabilistic Sentential Decision Diagram (PSDD)~\cite{conditional-psdd-sbn,structured-baysian-networks}. Despite these efforts, the scalability of decision diagrams for routing, in terms of space complexity, remains an open issue~\cite{psdd-structure-spaces-ChoiKRR15}.

The primary contribution of this work is to tackle the scalability challenges faced by the current state-of-the-art approaches. Our contributions are two-fold: first, we focus on minimizing the size of the compiled diagram by {\em relaxation and refinement}. In particular, instead of learning distributions over the set of all valid routes, we learn distributions over an over-approximation, perform sampling followed by refinement to output a valid route. Secondly, instead of a stepwise sampling procedure, we perform one-pass sampling by adapting existing sampling algorithm~\cite{YLM22} to perform conditional sampling. Our extensive empirical evaluation over benchmarks arising from real-world publicly available road network data demonstrates that our approach, called {\probroute}, is able to handle real-world instances that were clearly beyond the reach of the state-of-the-art. Furthermore, on instances that can be handled by prior state-of-the-art, {\probroute} achieves a median of $10\times$ runtime performance improvements.

\section{Background} \label{sec:bg}

In the remaining parts of this work we will discuss how to encode simple, more specifically simple trips, in a graph using Boolean formulas. In addition, we will also discuss decision diagrams and probabilistic reasoning with them. In this section, we introduce the preliminaries and background for the rest of the paper. 
To avoid ambiguity, we use \textit{vertices} to refer to nodes of road network graphs and \textit{nodes} to refer to nodes of decision diagrams. 

\subsection{Preliminaries} \label{sec:bg-prelim}

\paragraph{Simple Trip}
Let $G$ be an arbitrary undirected graph, a path on $G$ is a sequence of connected vertices $v_1, v_2, ... , v_m$ of $G$ where $\forall_{i=1}^{m-1} v_{i+1} \in N(v_i)$, with $N(v_i)$ referring to neighbours of $v_i$. A path $\pi$ does not contain loops if $\forall_{v_i, v_j \in \pi} v_i \neq v_j$. $\pi$ does not contain detour if $\forall_{v_i, v_j, v_k, v_l \in \pi} v_j \not \in N(v_i) \lor v_k \not \in N(v_i) \lor v_l \not \in N(v_i)$. Path $\pi$ is a simple path if it does not contain loops. A simple path $\pi$ is a simple trip if it does not contain detours. We denote the set of all simple trips in $G$ as $\simpletrip{G}$. In Figure \ref{fig:path-property}, d-e-h is a simple trip whereas d-e-f-c-b-e-h and d-e-f-i-h are not because they contain a loop and a detour respectively. We use $\mathsf{Term}(\pi)$ to refer to the terminal vertices of path $\pi$.

\begin{figure}[htb]
    \centering
    \begin{adjustbox}{width=0.22\columnwidth}
    \begin{tikzpicture}[
    roundnode/.style={circle, draw=black!60, very thick, minimum size=7mm},
    ]
    \node[roundnode](a) at (0, 0){$a$};
    \node[roundnode](b) at (1, 0){$b$};
    \node[roundnode](c) at (2, 0){$c$};
    \node[roundnode](d) at (0, -1){$d$};
    \node[roundnode](e) at (1, -1){$e$};
    \node[roundnode](f) at (2, -1){$f$};
    \node[roundnode](g) at (0, -2){$g$};
    \node[roundnode](h) at (1, -2){$h$};
    \node[roundnode](i) at (2, -2){$i$};

    \draw [-] (a) -- (b);
    \draw [-] (a) -- (d);
    \draw [-] (b) -- (c);
    \draw [-] (b) -- (e);
    \draw [-] (c) -- (f);
    \draw [-] (d) -- (e);
    \draw [-] (d) -- (g);
    \draw [-] (e) -- (h);
    \draw [-] (e) -- (f);
	\draw [-] (f) -- (i);
    \draw [-] (g) -- (h);
    \draw [-] (h) -- (i);

    \end{tikzpicture}
    \end{adjustbox}
    \caption{A 3 $\times$ 3 grid graph}
    \label{fig:path-property}
\end{figure}
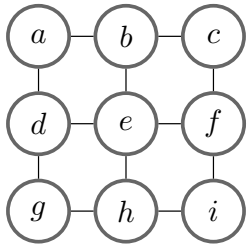 

\paragraph{Probabilistic Routing Problem}

In this paper, we tackle the probabilistic routing problem which we define as the following. Given a graph $G$ of an underlying road network, training and  testing data $D_{train}, D_{test}$, start and end vertices $s, t$, sample path $\pi$ from $s$ to $t$ such that $\varepsilon$-match rate with ground truth path $\pi' \in D_{test}$ is maximized. We define $\varepsilon$-match rate between $\pi$ and $\pi'$ as $|U_{close(\pi)}| \div |U|$ where $U$ is the set of vertices of $\pi'$ and $U_{close(\pi)}$ is the set of vertices of $\pi'$ that are within $\varepsilon$ euclidean distance away from the nearest vertex in $\pi$. More details on $\varepsilon$ will be discussed in Section~\ref{sec:experiments}.

\paragraph{Boolean Formula}

A Boolean variable can take values \textit{true} or \textit{false}. A literal $x$ is a Boolean variable or its negation. Let $F$ be a Boolean formula. $F$ is in conjunctive normal form (CNF) if $F$ is a conjunction of clauses, where each clause is a disjunction of literals. $F$ is satisfiable if there exists an assignment $\tau$ of the set of variables $X$ of $F$ such that $F$ evaluates to \textit{true}. We refer to $\tau$ as a satisfying assignment of $F$ and denote the set of all $\tau$ as $\mathsf{Sol}(F)$. In the remaining parts of this paper, all formulas and variables are Boolean unless stated otherwise.

\paragraph{Decision Diagrams}

Decision diagrams are directed acyclic graph (DAG) representations of logical formulas under the knowledge compilation paradigm. Decision diagrams are designed to enable the tractable handling of certain types of queries, that is the queries can be answered in polynomial time with respect to the number of nodes in the diagram~\cite{kc-map-darwiche}. We use diagram size to refer to the number of nodes in a decision diagram. In this work we use the OBDD[$\land$]~\cite{OBDDand-lai} variant of OBDD, more specifically Probabilistic OBDD[$\land$] ({\prob})~\cite{YLM22}, for which there are existing efficient sampling algorithm. We will discuss the {\prob} decision diagram in later sections.

\subsection{Related Works} \label{sec:related-works}

The continuous pursuit of compact representations by the research community has resulted in several decision diagram forms over the years. Some of the decision diagram forms include AOMDD for multi-valued variables, OBDD and SDD for binary variables~\cite{obdd-bryant,sdd,aomdd}. Both OBDD and SDD are canonical representations of Boolean formulas given variable ordering for OBDD and \textit{Vtree} for SDD respectively. OBDD~\cite{obdd-bryant} is comprised of internal nodes that correspond to variables and leaf nodes that correspond to $\top$ or $\bot$. Each internal node of OBDD have exactly two child and represents the Shannon decomposition~\cite{Boole54} on the variable represented by that internal node. SDDs are comprised of elements and nodes \cite{sdd}. Elements represent conjunction of a \textit{prime} and a \textit{sub}, each of which can either be a terminal SDD or a decomposition SDD. A decomposition SDD is represented by a node with child elements representing the decomposition. A terminal SDD can be a literal, $\top$ or $\bot$. The decompositions of SDDs follow that of the respective \textit{Vtree}, which is a full binary decision tree of Boolean variables in the formula. In this work, we use the OBDD[$\land$]~\cite{OBDDand-lai} variant of OBDD, which is shown to be theoretically incomparable but empirically more succinct than SDDs~\cite{OBDDand-lai}.

A related development formulates probabilistic circuits \cite{probabilistic-circuits-book}, based on sum-product networks \cite{spn} and closely related to decision diagrams, as a special class of neural networks known as Einsum networks \cite{einsum-networks}. In the Einsum network structure, leaf nodes represent different gaussian distributions. By learning from data, Einsum networks are able to represent SPDs as weighted sums and mixtures of gaussian distributions. Einsum networks address scalability by utilizing tensor operations implemented in existing deep learning frameworks such as PyTorch \cite{pytorch}. Our work differs from the Einsum network structure, we require the determinism property for decision diagrams whereas the underlying structure for Einsum network lacks this property. We will introduce the determinism property in the following section.

Various Boolean encodings have been proposed for representing paths within a graph, including Absolute, Compact, and Relative encodings~\cite{ham-path-encoding}. These encodings capture both the vertices comprising the path and the ordering information of said vertices. However, these encodings necessitate the use of polynomial number of Boolean variables, specifically $|V|^2$, $|V| log_{2} |V|$, and $2|V|^2$ variables for Absolute, Compact, and Relative encoding respectively. While these encodings accurately represent the space of paths within a graph, they are not efficient and lead to high space and time complexity for downstream routing tasks. 

Choi, Shen, and Darwiche demonstrated over a series of papers that the distribution of routes can be conceptualized as a structured probability distribution (SPD) given the underlying combinatorial structure~\cite{psdd-route,conditional-psdd-sbn,structured-baysian-networks}. This approach, referred to as the `CSD' approach in the rest of this paper, builds on top of existing approaches that represents paths using zero-surpressed decision diagrams~\cite{K05,Graphillion,KIIM17}. The CSD approach utilizes sentential decision diagrams to represent the SPD of paths and employs a stepwise methodology for handling path queries. Specifically, at each step, the next vertex to visit is determined to be the one with the highest probability, given the vertices already visited and the start and destination vertices. While the CSD approach has been influential in its incorporation of probabilistic elements in addressing the routing problem, it is not without limitations. In particular, there are two main limitations (1) there are no guarantees of completion, meaning that even if a path exists between a given start and destination vertex, it may not be returned using the CSD approach~\cite{psdd-route}. (2) the stepwise routing process necessitates the repeated computation of conditional probabilities, resulting in runtime inefficiencies. 

In summary, the limitations of prior works are Boolean encodings that require a high number of variables, lack of routing task completion guarantees, and numerous conditional probability computations.

\subsection{{\prob}: Probabilistic OBDD[$\land$]} \label{sec:notation-prelim}

In this subsection, we will introduce the {\prob} (Probabilistic OBDD[$\land$]) decision diagram structure and properties. We adopt the same notations as prior work~\cite{YLM22} for consistency.

\paragraph{Notations}

We use \textit{nodes} to refer to nodes in {\prob} $\psi$ and \textit{vertices} to refer to nodes in graph $G(V,E)$ to avoid ambiguity. $\child{n}$ refers to the children of node $n$. $\hi{n}$ refers to the hi child of decision node $n$ and $\lo{n}$ refers to the lo child of $n$. $\hiparam{n}$ and $\loparam{n}$ refer to the parameters associated with the edge connecting decision node $n$ with $\hi{n}$ and $\lo{n}$ respectively in a {\prob}. 
$\var{n}$ refers to the associated variable of decision node $n$. $\varset{n}$ refers to the set of variables of $F$ represented by a {\prob} with $n$ as the root node. $\subdiagram{n}$ refers to the {\prob} starting at node $n$. $\parent{n}$ refers to the immediate parent nodes of $n$ in {\prob}.

\paragraph{{\prob} Structure}

Let $\psi$ be a {\prob} which represents a Boolean formula $F$. $\psi$ is a DAG comprising of four types of nodes - conjunction node, decision node, true node, and false node.

A conjunction node (or $\land$-node) $n_c$ splits Boolean formula $F$ into different sub-formulas, i.e. $F = F_1 \land F_2 \land ... \land F_j$. Each sub-formula is represented by a {\prob} rooted at the corresponding child node of $n_c$, such that the set of variables in each of $F_1$, $F_2$, ..., $F_j$ are mutually disjoint.

A decision node $n_d$ corresponds to a Boolean variable $x$ and has exactly two child nodes, $\hi{n_d}$ and $\lo{n_d}$. $\hi{n_d}$ represents the decision to set $x$ to \textit{true} and $\lo{n_d}$ represents otherwise. We use $\var{n_d}$ to refer to the Boolean variable $x$ that decision node $n_d$ is associated with in $F$. Each branch of $n_d$ has an associated parameter, and the branch parameters of $n_d$ sum up to 1.

The leaf nodes of {\prob} $\psi$ are true and false nodes. An assignment $\tau$ of Boolean formula $F$ is a traversal of the {\prob} from the root node to the leaf node, we denote such a traversal as $\mathsf{Rep}_{\psi} (\tau)$. At each decision node $n_d$, the traversal follows the value of variable $\var{n_d}$ in $\tau$. At each conjunction node, all child branches are traversed. A satisfying assignment of $F$ will result in a traversal from root to leaf nodes where only the true nodes are visited. If a traversal leads to a false node at any point, then the assignment is not a satisfying assignment. 

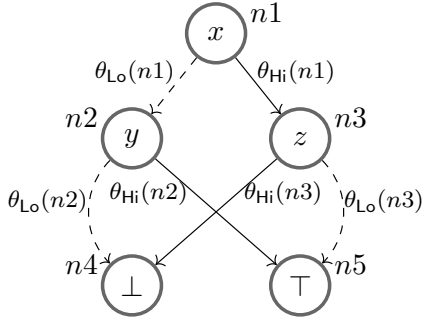
\begin{figure}[htb]
	\centering
    \begin{adjustbox}{width=0.4\columnwidth}
	\begin{tikzpicture}[
	roundnode/.style={circle, draw=black!60, very thick, minimum size=7mm},
	]
	\node[roundnode](x) at (0, 0){$x$};
	\node[roundnode](y) at (-1, -1.25){$y$};
	\node[roundnode](z) at (1, -1.25){$z$};
	\node[roundnode](t) at (1, -3){$\top$};
	\node[roundnode](f) at (-1, -3){$\bot$};

  \node(n1) at (0.6,0.25){$n1$};
  \node(n2) at (-1.6,-1.0){$n2$};
  \node(n3) at (1.6,-1.0){$n3$};
  \node(n3) at (-1.6,-2.75){$n4$};
  \node(n3) at (1.6,-2.75){$n5$};

  \node(n2hi) at (-0.8, -1.9){\footnotesize $\hiparam{n2}$};
  \node(n2lo) at (-2.0, -2.0){\footnotesize $\loparam{n2}$};
  \node(n3hi) at (0.8, -1.9){\footnotesize $\hiparam{n3}$};
  \node(n3lo) at (2.0, -2.0){\footnotesize $\loparam{n3}$};

  \draw [dashed,->] (x) -- (y) node[near start,left,rotate=0] {\footnotesize $\loparam{n1}$};
  \draw [->] (x) -- (z) node[near start,right,rotate=0] {\footnotesize $\hiparam{n1}$};
  \draw [->] (y) -- (t) ;
  \draw [dashed,->] (y) to[out=-135,in=135] (f);
  \draw [dashed,->] (z) to[out=-45,in=45] (t);
  \draw [->] (z) -- (f);

  \end{tikzpicture}
  \end{adjustbox}
  \caption{A {\prob} $\psi_1$ representing $F = (x \lor y) \land (\neg x \lor \neg z)$}
  \label{fig:Prob-non-smooth}
\end{figure}

An assignment of Boolean formula $F$ is represented by a top-down traversal of a {\prob} compiled from $F$. For example, we have a Boolean formula $F = (x \lor y) \land (\neg x \lor \neg z)$, represented by the {\prob} $\psi_1$ in Figure \ref{fig:Prob-non-smooth}. When $x$ is assigned \textit{true} and $z$ is assigned \textit{false}, $F$ will evaluate to \textit{true}. 
If we have a partial assignment $\tau$, we can perform inference conditioned on $\tau$ if we visit only the branches of decision nodes in $\psi$ that agree with $\tau$. This allows for conditional sampling, which we discuss in Section \ref{sec:sampling-approach}.

{\prob} inherits important properties of OBDD[$\land$] that are useful to our algorithms in later sections. The properties are - \textit{determinism}, \textit{decomposability}, and \textit{smoothness}. 

\begin{property}[Determinism]
    For every decision node $n_d$, the set of satisfying assignments represented by $\hi{n_d}$ and $\lo{n_d}$ are logically disjoint
\end{property}

\begin{property}[Decomposability]
	For every conjunction node $n_c$, $\varset{c_i} \cap \varset{c_j} = \emptyset, \forall c_i, c_j \in \child{n_c}, c_i \neq c_j$
\end{property}

\begin{property}[Smoothness]
    For every decision node $n_d$, $\varset{\hi{n_d}} = \varset{\lo{n_d}}$
\end{property}

In the rest of this paper, all mentions of {\prob} refer to smooth {\prob}. \textit{Smoothness} can be achieved via a smoothing algorithm introduced in prior work~\cite{YLM22}. We defer the smoothing algorithm to the appendix.

\section{Approach} \label{sec:sampling-approach}

In this section, we introduce our approach, {\probroute}, which addresses the aforementioned limitations of existing methods using (1) a novel relaxed encoding that requires a linear number of Boolean variables and (2) a single-pass sampling and refinement approach which provides completeness guarantees. The flow of {\probroute} is shown in Figure~\ref{fig:overall-flow}.

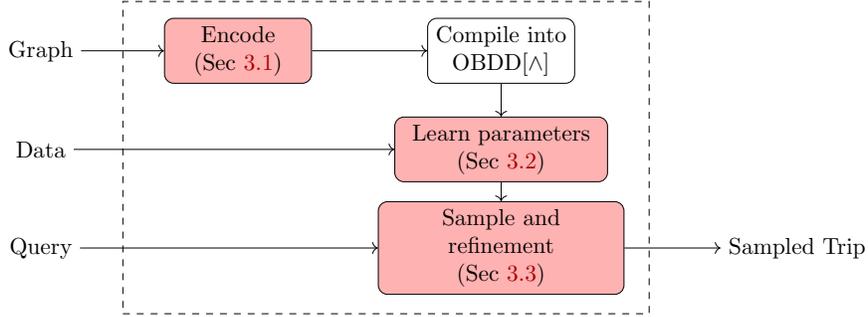
\begin{figure}[htb]
  \centering
  \begin{adjustbox}{width=0.8\columnwidth}
  \begin{tikzpicture}[
  node distance=2cm,
    diamondnode/.style={diamond, minimum width=1.5cm, minimum height=0.5cm, text centered, draw=black, fill=green!30},
    rectnode/.style={rectangle, rounded corners, minimum width=1.5cm, minimum height=0.5cm,text centered, draw=black, fill=red!30},
  rectnodenofill/.style={rectangle, rounded corners, minimum width=1.5cm, minimum height=0.5cm,text centered, draw=black},
  ]
  \node (encode) [rectnode, text width=2.0cm] at (0,0) {Encode\\(Sec \ref{subsec:encoding})};
  \node (compile) [rectnodenofill, text width=2cm] at (4, 0) {Compile into OBDD[$\land$]};
  \node (prob) [rectnode, text width=3cm] at (4, -1.5) {Learn parameters\\(Sec \ref{subsec:sampling-param})};
  \node (sample) [rectnode, text width=3.5cm] at (4, -3.0) {Sample and refinement\\(Sec \ref{subsec:sampling-algos})};

  \node (input-graph) at (-3, 0) {Graph};
  \node (input-data) at (-3, -1.5) {Data};
  \node (input-query) at (-3, -3) {Query};
  \node (output-path) at (8.5, -3) {Sampled Trip};

  \draw [->] (input-graph) -- (encode);
  \draw [->] (encode) -- (compile);
  \draw [->] (compile) -- (prob);
  \draw [->] (input-data) -- (prob);
  \draw [->] (input-query) -- (sample);
  \draw [->] (prob) -- (sample);
  \draw [->] (sample) -- (output-path);

  \draw[dashed] (-1.75, 0.75) rectangle (6.25,-4.00);

    \end{tikzpicture}
  \end{adjustbox}
    
  \caption{Flow of {\probroute}, with red rectangles indicating this work. For compilation, we use existing off-the-shelf techniques.}
  \label{fig:overall-flow}
\end{figure}

In our approach, we first use our relaxed encoding to encode $\simpletrip{G}$ of graph $G$ into a Boolean formula. Next, we compile the Boolean formula into OBDD[$\land$]. In order to learn from historical trip data, we convert the data into assignments. Subsequently, the OBDD[$\land$] is parameterized into {\prob} $\psi$ and the parameters are learned from data. Finally to sample trips from start vertex $v_s$ to destination vertex $v_t$, we create a partial assignment $\tau'$ with the variables that indicate $v_s$ and $v_t$ are terminal vertices set to \textit{true}. The {\probsample} algorithm, algorithm \ref{alg:ProbSample-traversal}, takes $\tau'$ as input and samples a set of satisfying assignments. Finally, in the refinement step, a simple trip $\pi$ is extracted from each satisfying assignment $\tau$ using depth-first search to remove disjoint loop components.

\subsection{Relaxed Encoding} \label{subsec:encoding}

In this work, we present a novel relaxed encoding that only includes vertex membership and terminal information. Our encoding only requires a linear ($2|V|$) number of Boolean variables, resulting in more succinct decision diagrams and improved runtime performance for downstream tasks. In relation to prior encodings, we observed that the ordering information of vertices can be inferred from the graph given a set of vertices and the terminal vertices, thus enabling us to exclude ordering information in our relaxed encoding. Our relaxed encoding represents an over-approximation of trips in $\simpletrip{G}$ for graph $G(V,E)$ using a linear number of Boolean variables with respect to $|V|$. We discuss the over-approximation in later parts of this section. 

Our encoding uses two types of Boolean variables, $n$-type and $s$-type variables. Each vertex $v \in V$ in graph $G(V,E)$ has a corresponding $n$-type and $s$-type variable. The $n$-type variable indicates if vertex $v$ is part of a trip and $s$-type variable indicates if $v$ is a terminal vertex of the trip. Our encoding is the conjunction of the five types of clauses over all vertices in graph $G$ as follows.

\begin{align}
  & \bigvee_{i \in V} s_{i} \label{cnf:at-least-one-s} \tag{H1}\\
  & \bigwedge_{i \in V} [n_{i} \longrightarrow \bigvee_{j \in adj(i)} n_{j}] \label{cnf:at-least-one-adj} \tag{H2}\\
  & \bigwedge_{\substack{i, j, k \in V, \\ i \not= j \not= k}} (\neg s_{i} \lor \neg s_{j} \lor \neg s_{k}) \label{cnf:at-most-two-s} \tag{H3}\\
  & \bigwedge_{i \in V} s_{i} \longrightarrow n_{i} \land \bigwedge_{j, k \in adj(i), j \not= k} (\neg n_{j} \lor \neg n_{k}) \label{cnf:terminal-at-least-one-adj} \tag{H4}\\
  & \bigwedge_{i,j \in V, j \in adj(i)} [n_i \land n_j \longrightarrow s_{i} \lor [(\bigvee_{\substack{k \in adj(i), \\ k \neq j \neq i}} n_k) \land \bigwedge_{\substack{l,m \in adj(i), \\ l,m \not= j}} (\neg n_l \lor \neg n_m )]] \label{cnf:terminal-or-exact-two-adj} \tag{H5}
\end{align}

A simple trip $\pi$ in graph $G$ has at least one terminal vertex and at most two terminal vertices, described by encoding components \ref{cnf:at-least-one-s} and \ref{cnf:at-most-two-s} respectively. At each terminal vertex $v_i$ of $\pi$, there can only be at most 1 adjacent vertex of $v_i$ that is also part of $\pi$ and this is encoded by \ref{cnf:terminal-at-least-one-adj}. For each vertex $v_i$ in $\pi$, at least one of their adjacent vertices is in $\pi$ regardless if $v_i$ is a terminal vertex or otherwise, this is captured by \ref{cnf:at-least-one-adj}. Finally, \ref{cnf:terminal-or-exact-two-adj} encodes that if a given vertex $v_i$ and one of its adjacent vertices are part of $\pi$, then either another neighbour vertex of $v_i$ is part of $\pi$ or $v_i$ is a terminal vertex.

\begin{definition} \label{def:path-to-sol}

  Let $\mathcal{M}: \simpletrip{G} \mapsto \sol{F}$ such that for a given trip $\pi \in \simpletrip{G}$, $\tau = \m{\pi}$ is the assignment whereby the $n$-type variables of all vertices $v \in \pi$ and the $s$-type variables of $v \in \terminal{\pi}$ are set to \textit{true}. All other variables are set to \textit{false} in $\tau$. 

\end{definition}

We refer to our encoding as relaxed encoding because the solution space of constraints over-approximates the space of simple trips in the graph. Notice that while all simple trips correspond to a satisfying assignment of the encoding, they are not the only satisfying assignments. Assignments corresponding to a simple trip $\pi$ with disjoint loop component $\beta$ also satisfy the constraints. The intuition is that $\beta$ introduces no additional terminal vertices, hence \ref{cnf:at-least-one-s}, \ref{cnf:at-most-two-s}, and \ref{cnf:terminal-at-least-one-adj} remain satisfied. Since the vertices in $\beta$ always have $n$-type variables of exactly two of its neighbours set to \textit{true}, \ref{cnf:terminal-or-exact-two-adj} and \ref{cnf:at-least-one-adj} remain satisfied. Thus, a simple trip with a disjoint loop component also corresponds to a satisfying assignment of our encoding.

\subsection{Learning Parameters from Data} \label{subsec:sampling-param}

\begin{algorithm}[htb]
  \begin{flushleft}
    \textbf{Input}: {\prob} $\psi$, $\tau$ - complete assignment of data instance\\
  \end{flushleft}
  \begin{algorithmic}[1] %
    \STATE $n \gets \mathsf{rootNode(\psi)}$
    \IF{$n$ is $\land$-node}
      \FOR{$c$ in $\child{n}$}
        \STATE $\mathsf{ProbLearn}(c, \tau)$ \label{line:problearn-conjunction-visit-all}
      \ENDFOR
    \ENDIF
    \IF{$n$ is decision node}
      \STATE $l \gets \mathsf{getLiteral(\tau, \var{n})}$ \label{line:problearn-decision-start}
      \IF{$l$ is positive literal}
        \STATE $\hicounter{n}$ += 1
        \STATE $\mathsf{ProbLearn(\hi{n}, \tau)}$
      \ELSE
        \STATE $\locounter{n}$ += 1
        \STATE $\mathsf{ProbLearn(\lo{n}, \tau)}$ \label{line:problearn-decision-end}
      \ENDIF
    \ENDIF
  \end{algorithmic}
  \caption{$\mathsf{ProbLearn}$ - updates counters of $\psi$ from data}
  \label{alg:prob-learn}
\end{algorithm}

We introduce algorithm~\ref{alg:prob-learn}, {\problearn}, for updating branch counters of {\prob} $\psi$ from assignments.
In order to learn branch parameters $\hiparam{n}$ and $\loparam{n}$ of decision node $n$, we require a counter for each of its branches, $\hicounter{n}$ and $\locounter{n}$ respectively. In the learning process, we have a dataset of assignments for Boolean variables in the Boolean formula represented by {\prob} $\psi$. For each assignment $\tau$ in the dataset, we perform a top-down traversal of $\psi$ following Algorithm~\ref{alg:prob-learn}. In the traversal, we visit all child branches of conjunction nodes (line~\ref{line:problearn-conjunction-visit-all}) and the child branch of decision node $n$ corresponding to the assignment of $\var{n}$ in $\tau$ (lines~\ref{line:problearn-decision-start} to~\ref{line:problearn-decision-end}), and increment the corresponding counters in the process. Subsequently, the branch parameters for node $n$ are updated according to the following formulas.

\begin{equation*}
  \hiparam{n} = \frac{\hicounter{n} + 1}{\hicounter{n} + \locounter{n} + 2}  \qquad \hfill \loparam{n} = \frac{\locounter{n} + 1}{\hicounter{n} + \locounter{n} + 2}
\end{equation*}

While we add 1 to numerator and 2 to denominator as a form of Laplace smoothing~\cite{MRS08}, other forms of smoothing to account for division by zero is possible. Notice that the learnt branch parameters of node $n$ are in fact approximations of conditional probabilities according to Proposition~\ref{prop:conditional-sample} and Remark~\ref{remark:app-learnt-conditional-prob} as follows.

\begin{proposition} \label{prop:conditional-sample}

  Let $n1$ and $n2$ be decision nodes where $n1 = {\parent{n2}}$ and $\lo{n1} = n2$, $\loparam{n2} = \frac{\locounter{n2} + 1}{\locounter{n1} + 2}$ and $\hiparam{n2} = \frac{\hicounter{n2} + 1}{\locounter{n1} + 2}$.

\end{proposition}

\begin{proof}

  Recall that the Lo branch parameter of $n2$ is:

  \begin{align*}
    \loparam{n2} = \frac{\locounter{n2} + 1}{\hicounter{n2} + \locounter{n2} + 2} 
  \end{align*}

  Notice that $\hicounter{n2} + \locounter{n2} = \locounter{n1}$ as all top-down traversals of $\psi$ that pass through $n2$ will have to pass through the Lo branch of $n1$.

  \begin{align*}
    \loparam{n2} = \frac{\locounter{n2} + 1}{\locounter{n1} + 2} 
  \end{align*}
  A similar argument can be made for $\hiparam{n2}$ by symmetry. In the general case if $n2$ has more than one parent, then the term $\hicounter{n2} + \locounter{n2}$ is the sum of counts of branch traversals of all parent nodes of $n2$ that leads to $n2$. Additionally, any conjunction node $c$ between $n1$ and $n2$ will not affect the proof because all children of $c$ will be traversed. For understanding, one can refer to the example in Figure \ref{fig:Prob-non-smooth} where $n1$ corresponds to the root node.

\end{proof}

\begin{remark} \label{remark:app-learnt-conditional-prob}
	Recall that $\var{n1} = x$ and $\var{n2} = y$ in {\prob} $\psi_1$ in Figure~\ref{fig:Prob-non-smooth}. Observe that $\frac{\locounter{n2}}{\locounter{n1}}$ for {\prob} $\psi_1$ in Figure~\ref{fig:Prob-non-smooth} is the conditional probability $Pr(\neg y | \neg x)$ as it represents the count of traversals that passed through Lo branch of $n2$ out of total count of traversals that passed through Lo branch of $n1$. A similar observation can be made for $\hi{n2}$.

  Notice that as the $\locounter{n2}$ and $\locounter{n1}$ becomes significantly large, that is $\locounter{n2} >> 1$ and $\locounter{n1} >> 2$:
  \begin{align*}
    \loparam{n2} = \frac{\locounter{n2} + 1}{\locounter{n1} + 2} \approx \frac{\locounter{n2}}{\locounter{n1}} = Pr(\neg y | \neg x)
  \end{align*}

  As such, the learnt branch parameters are approximate conditional probabilities. 
\end{remark}

\subsection{Sampling Trip Query Answers} \label{subsec:sampling-algos}

\begin{algorithm}[htb]
  \begin{flushleft}
    \textbf{Input}: {\prob} $\psi$, partial assignment $\tau'$\\
    \textbf{Output}: complete assignment that agrees with $\tau'$
  \end{flushleft}
  \begin{algorithmic}[1] %
    \STATE caches $\omega, \gamma$ $\longleftarrow$ initCache()
    \FOR{node $n$ in bottom-up ordering of $\psi$}
      \IF{$n$ is $\top$ node} 
        \STATE $\omega[n] \longleftarrow \emptyset$, $\gamma[n] \longleftarrow 1$
      \ELSIF{$n$ is $\bot$ node}
        \STATE $\omega[n] \longleftarrow$ $\mathsf{Invalid}$, $\gamma[n] \longleftarrow 0$
      \ELSIF{$n$ is $\land$ node}
        \STATE $\omega[n] \longleftarrow$ $\mathsf{unionChild}$(${\child{n}}, \omega$) \label{line:probsample-union-child}
        \STATE $\gamma[n] \longleftarrow \prod_{c \in \child{n}} \gamma[c]$
      \ELSE
        \IF{${\var{n}}$ in $\tau'$} \label{line:probsample-conditioning-start}
          \IF{$\omega[\tau'[{\var{n}}]]$ is $\mathsf{Invalid}$}
            \STATE $\omega[n] \longleftarrow$ $\mathsf{Invalid}$, $\gamma[n] \longleftarrow 0$
          \ELSE 
            \STATE $\omega[n] \longleftarrow$ $\mathsf{followAssign}$($\tau$)
            \IF{$\tau'[{\var{n}}]$ is $\neg \var{n}$}
              \STATE $\gamma[n] \longleftarrow \loparam{n} \times \gamma[\lo{n}]$
            \ELSE
              \STATE $\gamma[n] \longleftarrow \hiparam{n} \times \gamma[\hi{n}]$
            \ENDIF
          \ENDIF \label{line:probsample-conditioning-end}
        \ELSE
          \STATE $l \longleftarrow \loparam{n} \times \gamma[\lo{n}]$ \label{line:prob-sample-unassigned-start}
          \STATE $h \longleftarrow \hiparam{n} \times \gamma[\hi{n}]$
          \STATE $\gamma[n] \longleftarrow l + h$
          \STATE $\alpha \longleftarrow \mathsf{Binomial}(\frac{h}{l + h})$
          \IF{$\alpha$ is 1} 
            \STATE $\omega[n] \longleftarrow \omega[\hi{n}] \cup \var{n}$ 
          \ELSE
            \STATE $\omega[n] \longleftarrow \omega[\lo{n}] \cup \neg \var{n}$ \label{line:prob-sample-unassigned-end}
          \ENDIF
        \ENDIF
      \ENDIF
    \ENDFOR
    \STATE \textbf{return} $\omega$[rootnode($\psi$)]
  \end{algorithmic}
  \caption{{\probsample} - returns sampled assignment}
  \label{alg:ProbSample-traversal}
\end{algorithm}

The ability to conditionally sample trips is critical to handling trip queries for arbitrary start-end vertices, for which a trip is to be sampled conditioned on the given start and end vertices. In this work, we adapted the weighted sampling algorithm using {\prob}, which was introduced by prior work~\cite{YLM22}, to support conditional sampling and denote it as {\probsample}.

Algorithm \ref{alg:ProbSample-traversal}, {\probsample}, performs sampling of satisfying assignments from a {\prob} $\psi$ in a bottom-up manner. {\probsample} takes an input {\prob} $\psi$ and partial assignment $\tau'$ and returns a sampled complete assignment that agrees with $\tau'$. The input $\tau'$ specifies the terminal vertices for a given trip query by assigning the $s$-type variables. {\probsample} employs two caches $\omega$ and $\gamma$, for partially sampled assignment at each node and joint probabilities during the sampling process. In the process, {\probsample} performs calculations of joint probabilities at each node. In addition, {\probsample} stores the partial samples at each node in $\omega$. The partial sample for a false node would be $\mathsf{Invalid}$ as it means that an assignment is unsatisfiable. On the other hand, the partial sample for a true node is $\emptyset$ which will be incremented with variable assignments during the processing of internal nodes of $\psi$. The partially sampled assignment at every $\land$-node $c$ is the union of the samples of all its child nodes, as the child nodes have mutually disjoint variable sets due to \textit{decomposability} property. For a decision node $d$, if ${\var{d}}$ is in $\tau'$, the partial sample at $d$ will be the union of the literal in $\tau'$ and the partial sample at the corresponding child node (lines~\ref{line:probsample-conditioning-start} to \ref{line:probsample-conditioning-end}) to condition on $\tau'$. Otherwise, the partial assignment at $d$ is sampled according to the weighted joint probabilities $l$ and $h$ (lines~\ref{line:prob-sample-unassigned-start} to~\ref{line:prob-sample-unassigned-end}). Finally, the output of {\probsample} would be the sampled assignment at the root node of $\psi$. To extend {\probsample} to sample $k$ complete assignments, one has to keep $k$ partial assignments in $\omega$ at each node during the sampling process and sample $k$ independent partial assignments at each decision node. 

\begin{proposition} \label{prop:sampling-chain-branch-param}
  Let {\prob} $\psi$ represent Boolean formula $F$, {\probsample} samples $\tau \in \sol{F}$ according to the joint branch parameters, that is $\prod_{n \in \rep{\psi}{\tau}} [(1 - I_n)\loparam{n} + I_n \hiparam{n}]$ where $I_n$ is 1 if $\hi{n} \in \rep{\psi}{\tau}$ and 0 otherwise.
\end{proposition}

\begin{proof}
  Let $\psi$ be a {\prob} that only consists of decision nodes as internal nodes. At each decision node $d$ in the bottom-up sampling pass, assignment of $\var{d}$ is sampled proportional to $\loparam{d} \times \gamma[\lo{d}]$ and $\hiparam{d} \times \gamma[\hi{d}]$ to be \textit{false} and \textit{true} respectively. Without loss of generality, we focus on the term $\loparam{d} \times \gamma[\lo{d}]$, a similar argument would follow for the other branch by symmetry. 

  Let $d2$ denote $\lo{d}$. Notice that $\gamma[d2]$ is $\loparam{d2} \times \gamma[\lo{d2}] + \hiparam{d2} \times \gamma[\hi{d2}]$. Expanding the equation, the probability of sampling $\neg \var{d}$ is $\loparam{d} \times \loparam{d2} \times \gamma[\lo{d2}] + \loparam{d} \times \hiparam{d2} \times \gamma[\hi{d2}]$. If we expand $\gamma[\lo{d}]$ recursively, we are considering all possible satisfying assignments of $\varset{\lo{d}}$, more specifically we would be taking the sum of the product of corresponding branch parameters of each possible satisfying assignment of $\varset{\lo{d}}$.

  Observe that $\var{d}$ is sampled to be assigned \textit{false} with probability $\loparam{d} \times \loparam{d2} \times \gamma[\lo{d2}] + \loparam{d} \times \hiparam{d2} \times \gamma[\hi{d2}]$. Similarly, $\var{d2}$ is sampled to be assigned \textit{false} with probability $\loparam{d2} \times \gamma[\lo{d2}]$. Notice that if we view the bottom-up process in reverse, the probability of sampling $\neg \var{d}$ and $\neg \var{d2}$ is $\loparam{d} \times \loparam{d2} \times \gamma[\lo{d2}]$. In the general case, it then follows that a satisfying assignment would reach the \textit{true} node which has $\gamma$ value set to 1. It then follows that for each $\tau \in \mathsf{Sol}(F)$, $\tau$ is sampled with probability $P = \prod_{n \in \mathsf{Rep}_{\psi}(\tau)} [(1 - I_n) \loparam{n} + I_n \hiparam{n}]$. Notice that $\land$-nodes have no impact on the sampling probability as no additional terms are introduced in the product of branch parameters.
\end{proof}

\begin{remark}
  Recall in Remark~\ref{remark:app-learnt-conditional-prob} that $\hiparam{n}$ and $\loparam{n}$ are approximately conditional probabilities. By Proposition~\ref{prop:sampling-chain-branch-param}, assignment $\tau \in \sol{F}$ is sampled with probability proportional to $\prod_{n \in \mathsf{Rep}_{\psi}(\tau)} [(1 - I_n) \loparam{n} + I_n \hiparam{n}]$. Notice that if we rewrite the product of branch parameters as the product of approximate conditional probability, it is approximately the joint probability of sampling $\tau$.
\end{remark} 

\paragraph{Refinement}
In the refinement step, we extract a trip from sampled assignment $\tau$ by removing spurious disjoint loop components using depth-first search. 

\begin{definition} \label{def:sol-to-path}
  Let $\mathcal{M'}: \sol{F} \mapsto \simpletrip{G}$ be the mapping function of the refinement process, for a given graph $G$ and its relaxed encoding $F$. 
  For an assignment $\tau \in \sol{F}$, let $V_{\tau}$ be the set of vertices in $G$ that have their $n$-type variables set to \textit{true} in $\tau$. A depth-first search is performed from the starting vertex on $V_{\tau}$, removing disjoint components. The resultant simple path is $\pi = \invm{G}$.
\end{definition}

Although $\mathcal{M'}(\cdot)$ is a many-to-one (i.e. surjective) mapping function, it is not a concern in practice as trips with disjoint loop components are unlikely to occur in real-world or synthetic trip data from which probabilities can be learned. 

\begin{theorem}\label{prop:sampling-completeness}
  Given $v_s, v_t \in G$, let $\pi_{s,t} \in \simpletrip{G}$ be a trip between $v_s$ and $v_t$. Let $R_{\pi_{s,t}} = \{\tau \mid (\tau \in \mathsf{Sol}(F)) \wedge  (\mathcal{M'}(\tau) = \pi_{s,t}) \}$. Then,

  \begin{equation*}
    \Pr[\pi_{s,t} \text{ is sampled}] \propto \sum\limits_{\tau \in R_{\pi_{s,t}}} \prod_{n \in \mathsf{Rep}_{\psi}(\tau)} [(1 - I_n)\loparam{n} + I_n \hiparam{n}]
  \end{equation*}
\end{theorem}

\begin{proof}
  From Definition~\ref{def:path-to-sol} and~\ref{def:sol-to-path}, one can say that given a graph $G$ and its relaxed encoding $F$, $\forall \pi \in \simpletrip{G}, \exists \tau \in \sol{F}$ such that $\invm{\tau} = \pi$.

  Notice that sampling $\pi_{s,t}$ amounts to sampling $\tau \in R_{\pi_{s,t}}$. As such, the probability of sampling $\pi_{s,t}$ would be the sum over probability of sampling each member of $R_{\pi_{s,t}}$. Recall that the probability of sampling a single assignment $\tau$ is proportional to $\prod_{n \in \mathsf{Rep}_{\psi}(\tau)} [(1 - I_n)\loparam{n} + I_n \hiparam{n}]$ by Proposition~\ref{prop:sampling-chain-branch-param}. As such the probability $Pr[\pi_{s,t} \text{ is sampled}]$ is proportional to $\sum_{\tau \in R_{\pi_{s,t}}} \prod_{n \in \mathsf{Rep}_{\psi}(\tau)} [(1 - I_n)\loparam{n} + I_n \hiparam{n}]$.
\end{proof}

\begin{remark}
  It is worth noting that $Pr[\pi_{s,t} \text{ is sampled}] > 0$, as all branch parameters are greater than 0 by definition. Recall that branch parameters are computed with a '+1' in numerator and '+2' in denominator, and given that branch counters are 0 or larger, branch parameters are strictly positive.
\end{remark}

\section{Experiments} \label{sec:experiments}

In order to evaluate the efficacy of {\probroute}, we built a prototype in Python 3.8 with NumPy~\cite{numpy}, toposort~\cite{toposort}, OSMnx\cite{B17}, and NetworkX~\cite{HSS08} packages. We employ KCBox tool\footnote{https://github.com/meelgroup/KCBox} for OBDD[$\land$] compilation~\cite{OBDDand-lai}. The experiments were conducted on a cluster of machines with Intel Xeon Platinum 8272CL processors and 64GB of memory. In the experiments, we evaluated {\probroute} against an adaptation of the state-of-the-art probabilistic routing approach~\cite{psdd-route} and an off-the-shelf non-probabilistic routing library, Pyroutelib3~\cite{pyroutelib}, in terms of quality of trip suggestions and runtime performance. In particular, we adapted the state-of-the-art approach by Choi et al~\cite{psdd-route} to sample for trips instead of computing the most probable trip and refer to the adapted approach as `CSD' in the rest of the section. In addition, we compared our relaxed encoding to existing path encodings across various graphs, specifically to absolute encoding and compact encoding~\cite{ham-path-encoding}.

Through the experiments, we investigate the following:
\begin{description}
  \item [R1] Can {\probroute} effectively learn from data and sample quality trips?
  \item [R2] How efficient is our relaxed encoding technique?
  \item [R3] What is the runtime performance of the {\probroute}?
\end{description}

\paragraph{Data Generation}

In this study, we use the real-world road network of Singapore obtained from OpenStreetMap~\cite{OpenStreetMap} using OSMnx. The road network graph $G_r$ consisted of 23522 vertices and 45146 edges along with their lengths. In addition, we also use an abstracted graph\footnote{We use the geohash system (geohash level 5) of partitioning the road network graph. For more information on the format http://geohash.org/site/tips.html\#format} of $G_r$ which we denote as $G_a$ for the remaining of this section. A vertex in $G_a$ corresponds to a unique subgraph of $G_r$.

Synthetic trips were generated by deviating from shortest path given start and end vertices. A random pair of start and end vertices were selected in $G_r$ and the shortest path $\pi$ was found. Next, the corresponding intermediate regions of $\pi$ in $G_a$ are blocked in $G_r$, and a new shortest path $\pi'$ was found and deemed to be the synthetic trip generated. We generated 11000 synthetic trips, 10000 for training and 1000 for evaluation. While we used $G_a$ to keep the trip sampling time reasonable, it is possible to use more fine-grained regions for offline applications.

\subsection{R1: {\probroute}'s Ability to Learn Probabilities} \label{subsec:synthetic-rq}

To understand {\probroute}'s ability to learn probabilities from data, we evaluate its ability to produce trips that closely resembles the ground truth. Both {\probroute} and CSD, which are sampling-based approaches, were evaluated by sampling 20 trips and taking the median match rate for each instance. Recall that the $\varepsilon$-match rate is defined as the proportion of vertices in the ground truth trip that were within $\varepsilon$ meters of euclidean distance from the closest vertex in the proposed trip. In the evaluation, we set the $\varepsilon$ tolerance to be the median edge length of $G_r$, which is 64.359 meters, to account for parallel trips. To further emphasize the advantages of probabilistic approaches, we included an off-the-shelf routing library, Pyroutelib3~\cite{pyroutelib}, in the comparison.

In order to conduct a fair comparison, we have adapted the CSD approach to utilize {\prob} derived from our relaxed encoding. 
Our evaluation utilizes this adapted approach to sample a trip on $G_a$ in a stepwise manner, where the probability of the next step is conditioned on the previous step and destination. The conditional probabilities are computed in a similar manner to the computations of joint probabilities, which are the $\gamma$ cache values, in the {\probsample}. The predicted trip on the road network $G_r$ is determined by the shortest trip on the subgraph formed by the sequence of sampled regions. In contrast, {\probroute} samples a trip on $G_a$ in a single pass, and subsequently retrieves the shortest trip on the subgraph of the sampled regions as the predicted trip on $G_r$. It is worth noting that for sampling-based approaches, there may be instances where a trip cannot be found on $G_r$ due to factors such as a region in $G_a$ containing disconnected components. We incorporated a rejection sampling process with a maximum of 400 attempts and 5 minutes to account for such scenarios.

\begin{table}[htb]
  \centering
  \begin{small}
  \begin{tabular}{l|rrr|rrr}
  \toprule
  Stats & \multicolumn{3}{|c}{Exact Match} &\multicolumn{3}{|c}{$\varepsilon$-Match}\\
    &  Pyroutelib & CSD & {\probroute} & Pyroutelib & CSD & {\probroute}\\
  \midrule
  
  25\%  & 0.045 & 0.049 & \textbf{0.082} & 0.061 & 0.066 & \textbf{0.102} \\
  50\%  & 0.088 & 0.160 & \textbf{0.310} & 0.107 & 0.172 & \textbf{0.316} \\
  75\%  & 0.185 & 0.660 & \textbf{1.000} & 0.208 & 0.663 & \textbf{1.000} \\
  Mean  & 0.151 & 0.359 & \textbf{0.445} & 0.171 & 0.372 & \textbf{0.456} \\
  \bottomrule
  \end{tabular}
  \end{small}
  \caption{Match rate statistics for completed benchmark instances by respective methods. The percentages under `Stats' column refer to the corresponding percentiles. `Exact Match' refers to match rate when $\varepsilon=0$, and `$\varepsilon$-Match' refers to match rate when $\varepsilon$ is set to median edge length of $G_r$.}
  \label{tab:synthetic-stats}
\end{table}

Table~\ref{tab:synthetic-stats} shows the match rate statistics of the respective methods. Under $\varepsilon$-Match setting, where $\varepsilon$ is set as the median edge length of $G_r$ to account for parallel trips, {\probroute} has the highest match rate among the three approaches. In addition, {\probroute} produced perfect matches for more than 25\% of instances. {\probroute} has 0.316 $\varepsilon$-match rate on median, significantly more than 0.172 for CSD and 0.107 for Pyroutelib. The trend is similar for exact matches, where $\varepsilon$ is set to 0 as shown under the `Exact Match' columns in Table~\ref{tab:synthetic-stats}. In the exact match setting, {\probroute} achieved a median of 0.310 match rate, almost double that of CSD's 0.160 median match rate. The evaluation results also demonstrate the usefulness of probabilistic approaches such as {\probroute}, especially in scenarios where experienced drivers navigate according to their own heuristics which may be independent of the shortest trip. In particular, {\probroute} would be able to learn and suggest trips that align with the unknown heuristics of driver preferences given start and destination locations. Thus, the results provide an affirmative answer to \textbf{R1}.

\subsection{R2: Efficiency of Relaxed Encoding} \label{subsec:compilation-size}

\begin{table}[htb]
  \centering
  \begin{small}
  \begin{tabular}{l|rrrr|r}
  \toprule
  Encoding & \multicolumn{4}{c|}{Grid} & SGP\\
    &  $2$  & $3$    &  $4$   & $5$  & \\
  \midrule
  Absolute        & 99   & 1500   & 31768 & 1824769   & TO \\
  Compact         & 771   & TO    & TO    & TO        & TO \\
  Relaxed(Ours)  & \textbf{31}    & \textbf{146}   & \textbf{2368}   & \textbf{20030}    & \textbf{38318}\\
  \bottomrule
  \end{tabular}
  \end{small}
  \caption{Comparison of OBDD[$\land$] size for different graphs, with 3600s timeout. Grid 2 refers to a 2x2 grid graph. SGP refers to abstract graph ($G_a$) of Singapore road network.}
  \label{tab:encoding-diagram-size}
\end{table}

We compared our relaxed encoding to existing path encodings across various graphs, specifically to absolute encoding and compact encoding \cite{ham-path-encoding}. In the experiment, we had to adapt compact encoding to CNF form with Tseitin transformation~\cite{T83}, as CNF is the standard input for compilation tools. We compiled the CNFs of the encodings into OBDD[$\land$] form with 3600s compilation timeout and compared the size of resultant diagrams. The results are shown in Table~\ref{tab:encoding-diagram-size}, with rows corresponding to the different encodings used and columns corresponding to different graphs being encoded. Entries with {\em TO} indicate that the compilation has timed out. Table~\ref{tab:encoding-diagram-size} shows that our relaxed encoding consistently results in smaller decision diagrams, up to 91$\times$ smaller. It is also worth noting that relaxed encoding is the only encoding that leads to compilation times under 3600s for the abstracted \textit{Singapore} graph. The results strongly support our claims about the significant improvements that our relaxed encoding brings.

\subsection{R3: {\probroute}'s Runtime Performance} \label{subsec:runtime-comparison}

\begin{table}[htb]
  \centering
  \begin{small}
  \begin{tabular}{l|rr}
  \toprule
  Stats       & $\frac{\text{CSD}}{\text{Pyroutelib}} \times 10^3$   & $\frac{\text{{\probroute}}}{\text{Pyroutelib}} \times 10^3$ \\
  \midrule

  25\%        & 6.33    & \textbf{1.40}  \\
  50\%        & 21.64   & \textbf{2.00}  \\
  75\%        & 47.90   & \textbf{3.03}  \\
  Mean        & 36.16   & \textbf{2.62}  \\

  \bottomrule
  \end{tabular}
  \end{small}
  \caption{Relative runtime statistics (lower is better) for completed instances by CSD and {\probroute}. Under column `$\frac{\text{CSD}}{\text{Pyroutelib}}$' and row `50\%', CSD approach takes a median of $21.64 \times 10^3$ times the runtime of Pyroutelib.}
  \label{tab:runtime-comparison}
\end{table}

For wide adoption of new routing approaches, it is crucial to be able to handle the runtime demands of existing applications. As such, we measured the relative runtimes of probabilistic approaches, that is {\probroute} and CSD, with respect to existing routing system Pyroutelib and show the relative runtimes in Table~\ref{tab:runtime-comparison}. From the result, {\probroute} is more than one order of magnitude faster on median than the existing probabilistic approach CSD. The result also shows that {\probroute} is also on median more than a magnitude closer to Pyroutelib's runtime using the same {\prob} as compared to CSD approach. In addition, CSD approach timed out on 650 of the 1000 test instances, while {\probroute} did not time out. Additionally, as mentioned in~\cite{psdd-route}, CSD does not guarantee being able to produce a complete trip from start to destination. The results in Table~\ref{tab:runtime-comparison} highlight the progress made by {\probroute} in closing the gap between probabilistic routing approaches and existing routing systems.

\section{Conclusion} \label{sec:conclusion}

Whilst we have demonstrated the efficiency of our approach, there are possible extensions to make our approach more appealing for wide adoption. In terms of runtime performance, our approach is three orders of magnitude slower than existing probability agnostic routing systems. As such, there is still room for runtime improvements for our approach to be functional replacements of existing routing systems. Additionally, our relaxed encoding only handles undirected graphs at the moment and it would be of practical interest to extend the encoding to directed graphs to handle one-way streets. Furthermore, it would also be of interest to incorporate ideas to improve runtime performance from existing hierarchical path finding algorithms such as contractual hierarchies, multi-level dijkstra and other related works~\cite{MLD,Contraction-Hierarchies,ROAD}.

In summary, we focused on addressing the scalability barrier for reasoning over route distributions. To this end, we contribute two techniques: a relaxation and refinement approach that allows us to efficiently and compactly compile routes corresponding to real-world road networks, and a one-pass route sampling technique. We demonstrated the effectiveness of our approach on a real-world road network and observed around $91\times$ smaller {\prob}, $10\times$ faster trip sampling runtime and almost $2\times$ the match rate of state-of-the-art probabilistic approach, bringing probabilistic approaches closer to achieving comparable runtime to traditional routing tools. 

\section*{Acknowledgments} \label{sect:acks}

We sincerely thank Yong Lai for the insightful discussions. We sincerely thank reviewers for the constructive feedback. Suwei Yang is supported by the Grab-NUS AI Lab, a joint collaboration between GrabTaxi Holdings Pte. Ltd., National University of Singapore, and the Industrial Postgraduate Program (Grant: S18-1198-IPP-II) funded by the Economic Development Board of Singapore. Kuldeep S. Meel is supported in part by National Research Foundation Singapore under its NRF Fellowship Programme(NRF-NRFFAI1-2019-0004), Ministry of Education Singapore Tier 2 grant (MOE-T2EP20121-0011), and Ministry of Education Singapore Tier 1 Grant (R-252-000-B59-114).

\label{sect:bib}
\bibliographystyle{plain}
\bibliography{references.bib}

\clearpage
\appendix
\section*{Appendix}

\subsection*{Smoothing}

\begin{algorithm}[htb]
	\begin{flushleft}
		\textbf{Input}: {\prob} $\psi$\\
	\end{flushleft}
	\begin{algorithmic}[1] %
		\FOR{node $n$ in bottom-up pass of $\psi$}
			\IF{$n$ is $\top$ or $\bot$ node} 
				\STATE $\varset{n} \longleftarrow \emptyset$
			\ELSIF{$n$ is $\land$ node}
				\STATE $\varset{n} \longleftarrow \bigcup_{n_{c} \in \child{n}} \varset{n_{c}}$ 
			\ELSE
				\STATE $\delta \longleftarrow \mathsf{createDecisionNodes}(\varset{\lo{n}} \setminus \varset{\hi{n}})$ \label{app:line:smoothing-start}
				\STATE $c \longleftarrow \mathsf{conjunctionNode}()$
				\STATE $\child{c} \longleftarrow \{ \hi{n} \cup \delta \}$
				\STATE $\hi{n} \longleftarrow c$ \label{app:line:smoothing-end}

				\STATE $\delta \longleftarrow \mathsf{createDecisionNodes}(\varset{\hi{n}} \setminus \varset{\lo{n}})$
				\STATE $c \longleftarrow \mathsf{conjunctionNode}()$
				\STATE $\child{c} \longleftarrow \{ \lo{n} \cup \delta \}$
				\STATE $\lo{n} \longleftarrow c$
				\STATE $\varset{n} \longleftarrow {\varset{\lo(n)} \cup \varset{\hi{n}}}$
			\ENDIF
		\ENDFOR
	\end{algorithmic}
	\caption{$\mathsf{Smooth}$ - performs smoothing on a {\prob}}
	\label{app:alg:smooth}
\end{algorithm}

An important property to enable a {\prob} to learn the correct distribution is smoothness. A non-smooth {\prob} could be missing certain parameters. An example would be if we have an assignment $\tau_1 = [\neg x, y, -z]$, $\psi_1$ in Figure \ref{fig:Prob-non-smooth} (main paper) will not have a counter for $\neg z$ as the traversal ends after reaching the true node from the decision node representing variable $y$. Observe that the above-mentioned issue would not occur in a smooth {\prob} $\psi_2$ in Figure \ref{app:fig:smooth-prob-example}. If a {\prob} is not smooth, we can perform smoothing. For a given decision node $n$ (for consistency with algorithm~\ref{app:alg:smooth}), if $\varset{\lo{{n}}} \setminus \varset{\hi{{n}}} \neq \emptyset$ we can augment the hi branch of $n$ with the missing variables as shown in lines~\ref{app:line:smoothing-start} - \ref{app:line:smoothing-end} of algorithm~\ref{app:alg:smooth}. We create a new conjunction node $c$ with decision nodes representing each missing variable in $\varset{\lo{{n}}} \setminus \varset{\hi{{n}}}$ and $\hi{{n}}$ as children. For each additional decision node created for $\varset{\lo{{n}}} \setminus \varset{\hi{{n}}}$, both branches lead to the true node. We reassign $\hi{{n}}$ to be $c$. Similarly, we repeat the operation for $\lo{n}$. Once the smoothing operation is performed on every decision node in $\psi$, $\psi$ will have the smoothness property.

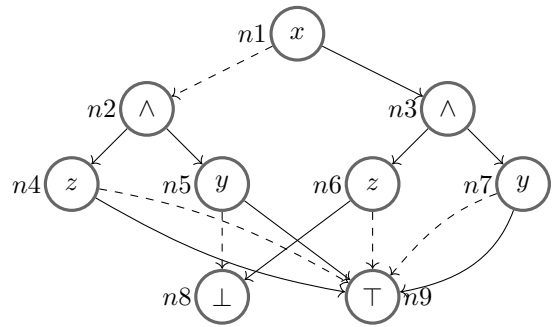
\begin{figure}[htb]
	\centering
	\begin{tikzpicture}[
	roundnode/.style={circle, draw=black!60, very thick, minimum size=7mm},
	]
	\node[roundnode](x) at (0, 0){$x$};
	\node[roundnode](left-x-and) at (-2, -1){$\land$};
	\node[roundnode](right-x-and) at (2, -1){$\land$};
	\node[roundnode](y) at (-1, -2){$y$};
	\node[roundnode](z) at (1, -2){$z$};
	\node[roundnode](z-c) at (-3, -2){$z$};
	\node[roundnode](y-c) at (3, -2){$y$};
	\node[roundnode](t) at (1, -3.5){$\top$};
	\node[roundnode](f) at (-1, -3.5){$\bot$};
	
	\draw [dashed,->] (x) -- (left-x-and);
	\draw [->] (x) -- (right-x-and);
	\draw [->] (y) -- (t);
	\draw [dashed,->] (y) -- (f);
	\draw [dashed,->] (z) -- (t);
	\draw [->] (z) -- (f);
	\draw [->] (left-x-and) -- (y);
	\draw [->] (left-x-and) -- (z-c);
	\draw [->] (right-x-and) -- (z);
	\draw [->] (right-x-and) -- (y-c);
	\draw [->] (z-c) to[out=-30,in=170] (t);
	\draw [dashed,->] (z-c) to[out=-10,in=150] (t);
	\draw [->] (y-c) to[out=-110,in=10] (t);
	\draw [dashed,->] (y-c) to[out=-160,in=50] (t);

	\node at (-0.6, 0) {$n1$};
	\node at (-2.6, -1){$n2$};
	\node at (1.4, -1){$n3$};
	\node at (-3.6, -2){$n4$};
	\node at (-1.6, -2){$n5$};
	\node at (0.4, -2){$n6$};
	\node at (2.4, -2){$n7$};
	\node at (-1.6, -3.5){$n8$};
	\node at (1.6, -3.5){$n9$};
	
	\end{tikzpicture}
	\caption{A smooth {\prob} $\psi_2$ with 9 nodes, $n1, ... , n9$, representing $F = (x \lor y) \land (\neg x \lor \neg z)$. Branch parameters are omitted}
	\label{app:fig:smooth-prob-example}
\end{figure}

\clearpage
\subsection*{Probability Computation}

\begin{algorithm}[h!]
	\begin{flushleft}
		\textbf{Input}: {\prob} $\psi$, Assignment $\tau$\\
		\textbf{Output}: probability of $\tau$
	\end{flushleft}
	\begin{algorithmic}[1] %
		\STATE cache $\gamma$ $\longleftarrow$ initCache()
		\FOR{node $n$ in bottom-up pass of $\psi$}
			\IF{$n$ is $\top$ node} 
				\STATE $\gamma[n] \longleftarrow 1$
			\ELSIF{$n$ is $\bot$ node}
				\STATE $\gamma[n] \longleftarrow 0$
			\ELSIF{$n$ is $\land$ node}
				\STATE $\gamma[n] \longleftarrow$ $\prod_{c \in {\child{n}}} \gamma[c]$ \label{line:probinfer-conjunction-node}
			\ELSE
				\IF{${\var{n}}$ in $\tau$}
					\STATE $\gamma[n] \longleftarrow$ $\mathsf{assignProb}$($\hiparam{n}, \loparam{n}, \gamma$) \label{line:probinfer-follow-assignment-decision}
				\ELSE
					\STATE $\gamma[n] \longleftarrow \loparam{n} \times \gamma[{\lo{n}}] + \hiparam{n} \times \gamma[{\hi{n}}]$ \label{line:probinfer-unassigned-decision}
				\ENDIF
			\ENDIF
		\ENDFOR
		\STATE \textbf{return} $\gamma$[rootnode($\psi$)]
	\end{algorithmic}
	\caption{{\probinfer} - returns probability of $\tau$}
	\label{alg:ProbInfer-traversal}
\end{algorithm}

Probability computations are typically performed on decision diagrams in a bottom-up manner, processing child nodes before parent nodes. In this work, we implement a probability computation algorithm ({\probinfer}), which is the joint probability computation component of {\probsample} in the main paper. In line~\ref{line:probinfer-follow-assignment-decision}, the $\gamma$ cache value is the product of branch parameter and child $\gamma$ cache value of the corresponding assignment of $\var{n}$ instead of both possible possible assignments in line~\ref{line:probinfer-unassigned-decision}.

\clearpage
\subsection*{Additional results}

We show in Table~\ref{app:tab:learning-stats} additional $\varepsilon$-match rate statistics on how well {\probroute} performs when provided different amount of data to learn probabilities. As we increase the amount of data provided for learning, in increments of 2000 instances (20\% of 10000 total), there is a general improvement in the match rate of the trips produced by {\probroute}. A similar trend is observed when $\varepsilon=0$, with corresponding stats shown in Table~\ref{app:tab:learning-stats-exact}.

\begin{table*}[htb]
    \centering
    \begin{small}
    \begin{tabular}{l|rrrrrr}
    \toprule
    Stats & \multicolumn{6}{c}{{\probroute}}\\
    & 0\% &20\% & 40\% & 60\% & 80\% & 100\% \\
    \midrule
    Mean  & 0.210 & 0.416 & 0.434 & 0.451 & 0.466 & 0.456 \\
    Std   & 0.192 & 0.360 & 0.373 & 0.376 & 0.383 & 0.386 \\
    25\%  & 0.081 & 0.102 & 0.095 & 0.098 & 0.098 & 0.102 \\
    50\%  & 0.149 & 0.286 & 0.297 & 0.318 & 0.349 & 0.316 \\
    75\%  & 0.257 & 0.715 & 0.854 & 0.964 & 1.000 & 1.000 \\
    \bottomrule
    \end{tabular}
    \end{small}
    \caption{$\varepsilon$-match rate statistics for {\probroute} where $\varepsilon$ is set as median edge length of road network graph $G_r$. The percentages under `Stats' column refer to the percentiles, for example `25\%' row  refer to the 25th percentile match rate for various methods. The percentages under {\probroute} header indicates the percentage of data that {\probroute} has learned from, out of the 10000 learning instances in total.}
    \label{app:tab:learning-stats}
\end{table*}

\begin{table*}[htb]
    \centering
    \begin{small}
    \begin{tabular}{l|rrrrrr}
    \toprule
    Stats & \multicolumn{6}{c}{{\probroute}}\\
    & 0\% &20\% & 40\% & 60\% & 80\% & 100\% \\
    \midrule
    Mean  & 0.192 & 0.404 & 0.422 & 0.440 & 0.455 & 0.445 \\
    Std   & 0.192 & 0.365 & 0.379 & 0.382 & 0.389 & 0.391 \\
    25\%  & 0.063 & 0.080 & 0.075 & 0.078 & 0.077 & 0.082 \\
    50\%  & 0.132 & 0.275 & 0.282 & 0.308 & 0.334 & 0.305 \\
    75\%  & 0.243 & 0.702 & 0.848 & 0.963 & 1.000 & 1.000 \\
    \bottomrule
    \end{tabular}
    \end{small}
    \caption{Exact match rate statistics for {\probroute} where $\varepsilon=0$. The percentages under `Stats' column refer to the percentiles, for example `25\%' row  refer to the 25th percentile match rate for various methods. The percentages under {\probroute} header indicates the percentage of data that {\probroute} has learned from, out of the 10000 learning instances in total.}
    \label{app:tab:learning-stats-exact}
\end{table*}

\end{document}